\documentclass[letterpaper, 10 pt, conference]{ieeeconf}  



\IEEEoverridecommandlockouts                              

\usepackage{cite}
\usepackage{amsmath,amssymb,amsfonts}
\newcommand{\norm}[1]{\left\lVert#1\right\rVert}
\usepackage{algorithmic}
\usepackage{url}
\usepackage[absolute,overlay]{textpos}
\usepackage{booktabs}
\usepackage[ruled,vlined,linesnumbered]{algorithm2e}
\usepackage{amsthm}
\newtheorem{theorem}{Theorem}
\newtheorem{proposition}{Proposition}

\usepackage{graphicx}
\usepackage{textcomp}
\usepackage{tabularx}
\usepackage{xcolor}

\newif\ifrevisions
\revisionsfalse   

\ifrevisions
    \newcommand{\rev}[1]{\textcolor{blue}{#1}}
\else
    \newcommand{\rev}[1]{#1}
\fi

\makeatletter
\def\ps@IEEEcopyrightfooter{%
  \def\@oddhead{}%
  \def\@evenhead{}%
  \def\@oddfoot{%
    \hfil
    \parbox{\textwidth}{%
      \centering
      \fbox{%
        \parbox{\dimexpr\textwidth-2\fboxsep-2\fboxrule\relax}{%
          \fontsize{8pt}{8.5pt}\selectfont
          \copyright 2026 IEEE. Personal use of this material is permitted. Permission from IEEE must be obtained for all other uses, in any current or future media, including reprinting/republishing this material for advertising or promotional purposes, creating new collective works, for resale or redistribution to servers or lists, or reuse of any copyrighted component of this work in other works.%
        }%
      }%
    }%
    \hfil
  }%
  \def\@evenfoot{\@oddfoot}%
}
\makeatother

\title{\LARGE \bf
From Global to Local: Hierarchical Probabilistic Verification for Reachability Learning
}

\author{Ebonye Smith$^{1}$, Sampada Deglurkar$^{1}$, Jingqi Li$^{2}$, Gechen Qu$^{1}$, and Claire J. Tomlin$^{1}$%
\thanks{*This work was supported by NSF Safe Learning Enabled Systems, 
the DARPA Assured Autonomy, ANSR, and TIAMAT programs, and
the NASA ULI on Safe Aviation Autonomy.  
E.S. was supported by NSF GRFP, and J.L. by an Oden
Institute Fellowship.
Any opinions, findings, and conclusions or recommendations expressed in this material are those of the authors and do not necessarily reflect the views of any aforementioned organizations. }
\thanks{$^{1}$Ebonye Smith (corresponding author), Sampada Deglurkar, Gechen Qu, and Claire J. Tomlin are with the Department of Electrical Engineering and Computer Sciences, University of California Berkeley, Berkeley, CA 94720, USA, tel:510-643-6610,  {\tt\small ebonyesmith@berkeley.edu, sampada\_deglurkar@berkeley.edu, qugch@berkeley.edu, tomlin@berkeley.edu}}%
\thanks{$^{2}$Jingqi Li is  with the Oden Institute of Computational Engineering and Sciences, University of Texas at Austin, Austin, TX 78712, USA, tel: 765-337-0678, {\tt\small jingqi.li@austin.utexas.edu}}%
}

\begin{document}

\maketitle
\thispagestyle{IEEEcopyrightfooter}
\pagestyle{empty}

\begin{abstract}
Hamilton–Jacobi (HJ) reachability provides formal safety guarantees for nonlinear systems. However, it becomes computationally intractable in high-dimensional settings, motivating learning-based approximations that may introduce unsafe errors or overly optimistic safe sets. In this work, we propose a hierarchical probabilistic verification framework for reachability learning that bridges offline global certification and online local refinement. We first construct a coarse\rev{, conservative under-approximation of the reach-avoid} safe set using scenario optimization, providing an efficient global probabilistic certificate. We then introduce an online local refinement module that \rev{is structurally anchored to the global certificate}
and recovers regions excluded by the global verification. This refinement reduces conservatism while focusing computation on critical regions of the state space. Integrated with a \rev{condition-based} switching mechanism between a learned reachability policy and a model-based controller, the proposed framework improves success rates in goal-reaching tasks with safety constraints, as demonstrated in \rev{drone-racing} simulation experiments.
\end{abstract}


\section{Introduction}
\label{introduction}

Provably safe controller design is essential for deploying autonomous systems in safety-critical settings such as multi-agent coordination, human–robot interaction, and high-speed navigation in cluttered environments. Among those applications, a central challenge is to characterize the set of states from which a system can be driven to a desired goal while satisfying safety constraints \cite{wabersich2023data}. Hamilton–Jacobi (HJ) reachability \cite{tomlin2002conflict,lygeros2004reachability,mitchell2005time} computes such sets, but classical HJ methods suffer from the curse of dimensionality \cite{bansal2017hamilton}. To overcome this issue, prior works have proposed learning-based approaches that approximate reachable sets and synthesize corresponding control policies \cite{bansal2021deepreach,fisac2019bridging,hsu2021safety,li2025certifiable}. While these approaches enable scalability, they may introduce errors that lead to overly optimistic approximations that misclassify unsafe regions as safe.

\begin{figure}[t]
    \centering
    \includegraphics[width=0.9\linewidth]{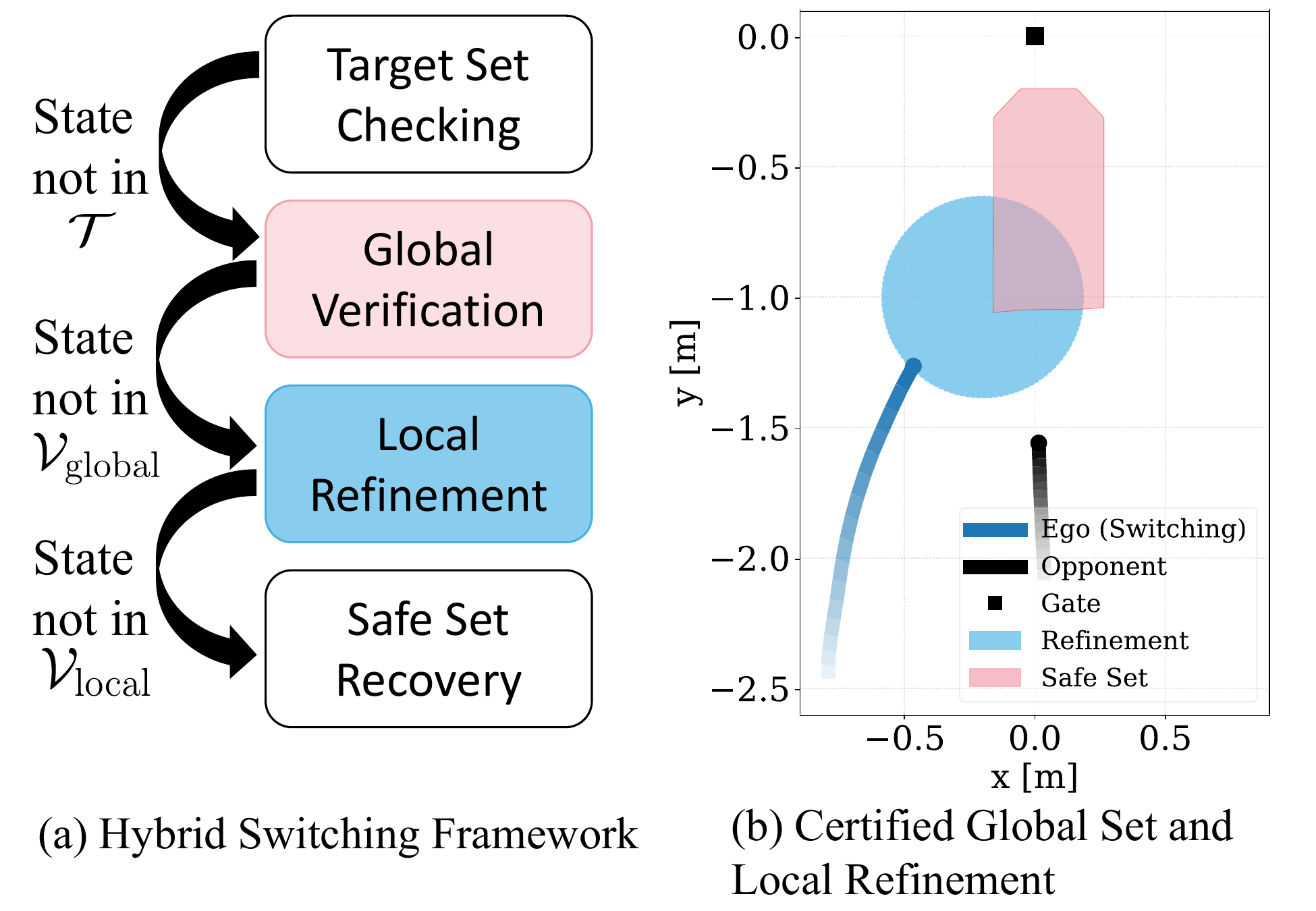}\vspace{-0.5em}
    \caption{
    (a) Building on a learned reachability value function and policy from prior work (e.g., \cite{li2025certifiable}), our hierarchical certification-guided framework \rev{uses probabilistically certified global and local safe sets to guide decision-making, invoking} \emph{safe set recovery} to drive it to a certified region. Our key novelty is leveraging scenario optimization for local safe set expansion. (b) Drone racing example in which the ego drone (blue) reaches a refined locally certified safe set which empirically results in safe overtaking of an opponent (black).
    }
    \label{fig:cover_fig}
\end{figure}
 
\rev{To improve reliability} under learning-based approximations, verification becomes necessary. Existing methods can be categorized along three axes. First, deterministic approaches provide worst-case guarantees but can be overly conservative in practice \cite{althoff2011zonotope,coogan2020mixed,hu2020reach,lew2021sampling,abate2022robustly,li2025certifiable}, whereas probabilistic methods are typically less conservative but may require many samples to achieve tight guarantees \cite{liebenwein2018sampling,devonport2020estimating,cao2021estimating,devonport2021data,lin2023generating,lindemann2023safe,lin2024verification,lin2025robust,dietrich2025data}. 
Second, verification acting \textit{globally} over the entire state space is computationally expensive and sensitive to outliers \cite{lin2025robust}, while \textit{local} verification focuses on regions of interest, enabling more sample-efficient certification. 
Third, due to its high computational cost, global verification is typically performed \emph{offline} \cite{lin2023generating}, whereas \textit{online} refinement enables adaptive correction of overly conservative certificates during execution.
Existing approaches typically operate within a single regime, limiting their ability to balance accuracy, efficiency, and scalability in high-dimensional systems.

In this work, we propose a \textit{hierarchical} \rev{certification-guided} framework, shown in Fig.~\ref{fig:cover_fig}, that couples global probabilistic certification with online local refinement 
\rev{for learning-based reachability}.
A coarse global certified 
region is first constructed via scenario optimization~\cite{campi2018introduction}, 
yielding a conservative but efficient \rev{probabilistic under-approximation 
of the reach-avoid set under the sampling distribution}.
An online local 
refinement module then performs targeted sampling near the global 
certificate boundary, expanding the certified region by recovering 
states excluded by the conservative global approximation.
\rev{The certified regions serve as 
decision-making guides: when the system lies within a certified region, 
the learned reachability policy $\pi_{lp}$ is trusted; otherwise, 
control transfers to a model-based fallback warm-started by $\pi_{lp}$.}
Our contributions are as follows:
\begin{itemize}
    \item \textbf{Hierarchical 
    \rev{certification-guided control} framework:} A multi-tier framework combining a global probabilistic
    \rev{certificate}
    with online local refinement \rev{and condition-based switching, using certified regions to guide when $\pi_{lp}$ can  be trusted versus when to defer to a fallback controller}.
    \item \textbf{\rev{Boundary-anchored local refinement of the probabilistically certified set:}}
    \rev{An online 
    scenario-based procedure that recovers probabilistically validated 
    states near the global certificate boundary.}
    \rev{\item \textbf{Scenario-based certification guarantees:}
    Theoretical guarantees characterizing the violation probability of the global and locally refined certified regions under the corresponding sampling distributions.}

    \item \textbf{Empirical validation:} Simulation results in a 12D drone racing case study demonstrating improved success rates in safe overtaking compared to baseline methods.
\end{itemize}

\section{Preliminaries}
\label{preliminaries}
In this section, we review reachability analysis and scenario optimization, which underpins our new probabilistic verification framework in Section~\ref{approach}.

\subsection{Reachability Learning}
\label{reachabilitylearning}
Consider a discrete-time nonlinear dynamical system $x_{t+1} = f(x_t , u_t)$,
where $x_t\rev{\in \mathcal{X}}\subseteq\mathbb{R}^n$ and $u_t\in \mathcal{U}\subseteq\mathbb{R}^m$ represent the state and control at time $t$. 
Let \(r:\mathbb{R}^{n}\to\mathbb{R}\) be a Lipschitz continuous reward function with Lipschitz constant \(L_r\), where \(r(x)>0\) indicates that a state \(x\) is in the target set \(\mathcal{T}\). 
Let \(c:\mathbb{R}^{n}\to\mathbb{R}\) be a Lipschitz continuous constraint function with Lipschitz constant \(L_c\), where \(c(x) \leq 0\) indicates constraint violation. Our goal is to identify the set of initial states from which the system can be safely controlled to $\mathcal{T}$ by control policy \(\pi:\mathbb{R}^n \to \mathcal{U}\) within a finite time horizon \(T \in \mathbb{Z}_+\), referred to as the \emph{reach-avoid (RA)} set \cite{margellos2011hamilton}:
\begin{equation}
\label{eq:RAset}
    \mathcal{R}:=\left\{
    \begin{array}{l}
        x_0:\exists\pi \hspace{2mm}\text{such that} \hspace{2mm}\exists t<T, \\
         \left(r(x_t)>0 \hspace{2mm}\wedge\hspace{2mm} \forall \tau \in [0,t], c(x_\tau)>0 \right) 
    \end{array}
    \right\}.
\end{equation}
Denote by \(\xi_{x_0}^{\pi}\) a state trajectory from a state $x_0$ under control policy $\pi$. We define the time-discounted \textit{RA measure} \(g_{\gamma}(\xi_{x_0}^{\pi},t)\) as
\begin{equation}
\label{eq:RAmeasure}
   g_{\gamma}(\xi_{x_0}^{\pi},t):=\min \Big\{\gamma ^tr(x_t), \min_{\tau=0,\dots,t} \gamma^{\tau} c(x_{\tau})\Big\},
\end{equation}
where $\gamma\in(0,1)$, and we have \(g_{\gamma}(\xi_{x_0}^{\pi},t)>0\) if and only if there exists a trajectory from \(x_0\) reaching the target set safely \rev{at time step $t$}.

Given a horizon $T$, the finite-horizon \textit{RA value function} \(V_{\gamma}(x,T)\) evaluates the maximum RA measure:
\begin{equation}
    V_{\gamma}(x_0, T):=\max_{\pi} \rev{\max_{t=0,\dots,T}} g_{\gamma}(\xi_{x_0}^{\pi}, t).
\end{equation}
For all \(\gamma \in (0,1)\), the super-zero level set of \(V_{\gamma}(x,T)\), defined as \(\mathcal{V}_{\gamma}:=\{x:V_{\gamma}(x,T)>0\}\), is equal to the RA set.

We leverage learning-based reachability methods (e.g., \cite{hsu2021safety,li2025certifiable}), which use deep reinforcement learning (RL) to learn a deterministic control policy \(\pi\) and value function \(V_{\gamma}\). In the rest of the paper, we denote the learned policy by \(\pi_{lp}\).

\subsection{Scenario Optimization Theory}
\label{scenarioopt}
Scenario optimization is a data-driven framework used to provide probabilistic guarantees for optimization problems under uncertainty \(\Delta\):
\begin{equation}
\begin{aligned}
    \min_{z \in \mathcal{Z}} \quad & c^\top z \\
    \text{s.t. } \quad & \mathbb{P}_{\Delta}\{h(z, \Delta) >0 \} \leq \epsilon
\end{aligned}
\label{eq:robust_opt}
\end{equation}
where \(z \in \mathcal{Z} \subseteq \mathbb{R}^d\) is the decision variable, \(h\) is convex, and \(\epsilon \in (0,1)\) is the maximum allowable risk.
To solve \eqref{eq:robust_opt} without assuming a specific distribution for \(\Delta\), we employ the \textit{scenario approach}, replacing the chance constraint with \(N\) deterministic constraints sampled from \(\Delta\):
\begin{equation}
\begin{aligned}
    z^\ast_N = & \arg\min_{z \in \mathcal{Z}} & & c^\top z \\
    & \text{s.t.} & & \bigwedge_{i=1}^N h(z, \delta^{(i)}) \leq 0
\end{aligned}
\label{eq:scenario_prob}
\end{equation}
where each $\delta^{(i)}$ is an i.i.d. realization of \(\Delta\). For user-specified risk level \(\epsilon\) and confidence parameter \(\beta\), the required number of deterministic constraints \(N\) to achieve a probabilistic guarantee is given in Theorem~\ref{samplecomplexitythm}.

\begin{theorem}[Sample Complexity \cite{campi2009scenario, devonport2020estimating}]
\label{samplecomplexitythm}

Let \(\epsilon\) be the risk level and \(\beta \in (0,1)\) be the confidence parameter.
Denote by \(\mathbb{P}_S\) the probability measure associated with the number of scenarios 
\(
    N \geq \frac{2}{\epsilon} \left( \ln \frac{1}{\beta} + d \right)
    \),
taken from \(\Delta\) and by \(\mathbb{P}_\Delta\) the probability measure associated with \(\Delta\). Then \(z_N^\ast\) satisfies:
\begin{equation}
\label{eq:scen_opt_ineq}
    \mathbb{P}_S(\mathbb{P}_\Delta(h(z_N^\ast, \delta)>0)\leq \epsilon)\geq 1-\beta.
\end{equation}
\end{theorem}
Inequality~\eqref{eq:scen_opt_ineq} states that with confidence at least 
$1-\beta$, the probability of an unseen $\delta \in \Delta$ violating 
$h$ is at most $\epsilon$. In practice, $\beta$ is set small 
(e.g.\ $0.001$) at low sample cost due to the log-inverse dependence, 
while $\epsilon$ is tuned to balance risk tolerance against sample 
complexity.


\section{Hybrid Reachability Verification Method}
\label{approach}

In this section, we describe our hierarchical verification framework\rev{\footnote{Code for this project can be found at: \url{https://github.com/ebonyelsmith/Hierarchical_Probabilistic_Verification}.}}.
\rev{In Section~\ref{dynamicsbounding}, we present a scenario optimization-based technique to bound the deviation of trajectories.}
\rev{Building on these trajectory divergence bounds, Section~\ref{subsec:global} details the construction of our global probabilistic certificate, while Section~\ref{subsec:local} introduces the online local refinement module used to expand this certificate near its boundary. Finally, Section~\ref{subsec:hierarchical} integrates these components into our complete hierarchical condition-based switching architecture.}

\subsection{Probabilistic Dynamics Bounding}
\label{dynamicsbounding}
\rev{To construct a global probabilistically certified under-approximation for the RA set, we must first characterize the sensitivity of the system dynamics under the learned policy \(\pi_{lp}\).}
The deterministic verification technique in \cite{li2025certifiable} leverages the Lipschitz constant of the dynamics.
Here, we alleviate conservatism by instead estimating a trajectory deviation quantity \(\Delta x_t^\ast\) using samples.


\rev{Let \(\bar{\mathcal{X}}\subseteq\mathcal{X}\) denote the nominal state sampling domain.} We collect \(N\) i.i.d. scenario pairs \(\{(\bar{x}^{(i)}_0, x^{(i)}_0)\}_{i=1}^N\) sampled uniformly, where each nominal state \(\bar{x}^{(i)}_0\) is sampled uniformly from \(\bar{\mathcal{X}}\), and each perturbed state \(x^{(i)}_0\) is sampled from a \(\rho\)-ball \(\mathcal{B}_{\rho}(\bar{x}^{(i)}_0)\). 
\rev{We denote the resulting pair space by} \(\mathcal{X}_{pair}\).
For each pair, we generate a nominal control sequence \(\mathbf{U}^{(i)} {=} \{{u}_0^{(i)}, \dots, {u}_{T-1}^{(i)}\}\) by applying \(\pi_{lp}\) to \(\bar{x}_0^{(i)}\), yielding the nominal trajectory \(\bar{\xi}^{(i)}{:=}\{\bar{x}_t\}_{t=0}^T\). 
For computational simplicity, we approximate the control for \(x_0^{(i)}\) using the same sequence \(\mathbf{U}^{(i)}\), yielding perturbed trajectory \(\xi^{(i)}:=\{{x}_t\}_{t=0}^T\).

The global sensitivity bound \(\Delta x_t^\ast\) is obtained by solving:
\begin{equation}
\begin{aligned}
    \min_{\Delta x_t \in \mathbb{R}} \quad \Delta x_t \\
    & \text{s.t. } \mathbb{P}_{{\mathcal{X}_{pair}}}\left(  \norm{x_t^{} - \bar{x}_t^{}}_2 > \Delta x_t \right) \leq \epsilon_{global}.
\end{aligned}
\label{eq:global_lipschitz_prob}
\end{equation}
We leverage scenario optimization theory to obtain the following scenario program:
\begin{equation}
\begin{aligned}
    \Delta x_t^\ast = & \min_{\Delta x_t \in \mathbb{R}} \quad \Delta x_t \\
    & \text{s.t.} \quad  \bigwedge_{i=1}^N \left(  \norm{x_t^{(i)} - \bar{x}_t^{(i)}}_2 \leq\Delta x_t \right)
\end{aligned}
\label{eq:global_lipschitz}
\end{equation}
By Theorem 1, the trajectory deviation remains bounded by \(\Delta x_t^\ast\) with probability \(1-\epsilon_{global}\) and confidence \(1-\beta_{global}\): 
\begin{equation}
    \mathbb{P}_{S_1}\left(\mathbb{P}_{{\mathcal{X}_{pair}}}\left(  \norm{x_t^{} - \bar{x}_t^{}}_2 > \Delta x_t^\ast \right) \leq \epsilon_{global} \right) \geq 1-\beta_{global}
\end{equation}
where \(\mathbb{P}_{\mathcal{S}_1}\) denotes \rev{the probability measure over the sampled scenario dataset used to construct the scenario program, while \(\mathbb{P}_{\mathcal{X}_{pair}}\) denotes the probability measure associated with the underlying pair sampling distribution over nominal and perturbed initial conditions. }


\subsection{Coarse Global Verification}
\label{subsec:global}

We leverage the trajectory deviation bound \(\Delta x_t^\ast\) derived in Section \ref{dynamicsbounding} to obtain a probabilistically certified lower bound of the value function, denoted as \(\check{V}_{\gamma}\), under learned policy \(\pi_{lp}\), as illustrated in Theorem \ref{globalsafetythm}.

\begin{theorem}[Global Probabilistic Safety]
\label{globalsafetythm}
Let $x_0\in \mathcal{B}_{\rho}(\bar{x}_0)$. Suppose that for each $t\in\{0,1,\dots,T\}$, $\Delta x_t^\ast$ is as calculated above.
Let $\check{r}_t:=r(\bar{x}_t)-L_r \Delta x_t^\ast $ and $\check{c}_t:=c(\bar{x}_t)-L_c \Delta x_t^\ast$. 
Define a certificate function $\check{V}_\gamma(\bar{x}_0, T)$, 
\begin{equation}\label{eq:lower bound of V_gamma}
    \check{V}_\gamma(\bar{x}_0, T):=\max_{t=0,\dots, T} \min \{\gamma^t \check{r}_t, \min_{\tau=0,\dots,t} \gamma^\tau \check{c}_\tau\}.
\end{equation}
Then, the following probabilistic guarantee holds:
\begin{equation*}
\begin{split}
    \mathbb{P}_{S_1}\Big(\mathbb{P}_{{\mathcal{X}_{pair}}}\big(\{V_\gamma(x_0,T) < 0\} \, \cap & \\
    \{\check{V}_\gamma(\bar{x}_0, T) \geq 0 \}\big) \leq \epsilon_{global}\Big) & \geq 1-\beta_{global}.
\end{split}
\end{equation*}
\end{theorem}

\begin{proof}
    Let $t\in\{1,\dots,T\}$. Similar to analysis in Section VI of \cite{li2025certifiable}, we define the set of states originating from the ball $\mathcal{B}_{\rho}(\bar{x}_0) $ under the nominal controls $\{\pi_{lp}(\bar{x}_\tau)\}_{\tau=0}^{t-1}$ as $\mathcal{X}_t:=\{x_t: x_{\tau+1}=f(x_{\tau}, \pi_{lp}(\bar{x}_\tau)), \tau\le t-1, x_0\in\mathcal{B}_{\rho}(\bar{x}_0) \}$. Observe that the set \(\bar{\mathcal{X}}_{t,\bar{x}_0}:=\{x_t: \norm{x_t^{} - \bar{x}_t^{}}_2 \leq \Delta x_t^\ast\}\) is a convex outer approximation of $\mathcal{X}_t$. By Lipschitz continuity of $r(x)$, we have that with confidence \(1-\beta_{global}\),
\begin{equation*}
    \forall x_t \in \bar{\mathcal{X}}_{t,\bar{x}_0}, \hspace{1mm} \mathbb{P}_{{\mathcal{X}_{pair}}}(|r(x_t) -r(\bar{x}_t)| \leq L_r \Delta x_t^\ast) \geq 1-\epsilon_{global}
\end{equation*}
which yields a lower bound of \(r(x_t)\), for all \(x_t \in \bar{\mathcal{X}}_{t,\bar{x}_0}\):
\begin{equation*}
    \mathbb{P}_{{\mathcal{X}_{pair}}}(\check{r}_t \leq r(x_t)) \geq 1-\epsilon_{global}.
\end{equation*}
Similarly, define a lower bound of \(c(x_t)\), for all \(x_t \in \bar{\mathcal{X}}_{t,\bar{x}_0}\):
\begin{equation*}
    \mathbb{P}_{{\mathcal{X}_{pair}}}(\check{c}_t \leq c(x_t))\geq 1-\epsilon_{global}.
\end{equation*}
As shown in \eqref{eq:lower bound of V_gamma}, we can use \(\check{r}_t\) and \(\check{c}_t\) to construct a probabilistic lower bound function \(\check{V}_\gamma (\bar{x}_0, T)\) for \(V_\gamma (x_0,T)\), for all \(x_0 \in \mathcal{B}_{\rho}(\bar{x}_0)\).
In other words, $\check{V}_\gamma(\bar{x}_0,T)$ serves as a lower bound of $ V_\gamma(x_0,T) $ with probability \(1-\epsilon_{global}\) over the support \({\mathcal{X}_{pair}}\). Denote by $X$ and $Y$ the events $ \{V_\gamma(x_0) <  \check{V}_\gamma(\bar{x}_0, T) \}$ and $\{\check{V}_\gamma(\bar{x}_0, T) \geq 0 \}$, respectively. Thus, we have with confidence at least \(1-\beta_{global}\), $\mathbb{P}_{{\mathcal{X}_{pair}}}(X) \leq \epsilon_{global} $. By the law of total probability, we have, with confidence at least $1-\beta_{global}$,
\begin{equation*}
    \mathbb{P}_{{\mathcal{X}_{pair}}}(X \cap Y) + \mathbb{P}_{{\mathcal{X}_{pair}}}(X \cap Y^c) \leq \epsilon_{global},
\end{equation*}
i.e., $\mathbb{P}_{{\mathcal{X}_{pair}}}(X \cap Y)\le \epsilon_{global}$. Thus, with confidence at least $1-\beta_{global}$,
\begin{equation}
    \mathbb{P}_{{\mathcal{X}_{pair}}}(\{V_\gamma(x_0,T) <  0\} \cap \{\check{V}_\gamma(\bar{x}_0, T) \geq 0 \} ) \leq \epsilon_{global}. \tag*{\qed}
\end{equation}
\renewcommand{\qed}{}
\end{proof}
Theorem \ref{globalsafetythm} states that, \rev{under the sampling distribution \(\mathcal{X}_{pair}\), the probability that the lower-bound certificate incorrectly classifies an unsafe state as safe is upper bounded by \(\epsilon_{global}\) with confidence \(1-\beta_{global}\).}
We define the 
\rev{globally certified region}
\(\mathcal{V}_{global}\)\rev{, which acts as a conservative under-approximation of the true RA set,} as:
\begin{equation}
    \mathcal{V}_{global} := \{ {x} \in \mathcal{\bar{X}} \mid \check{V}_{\gamma}(x, T) \geq 0 \}.
    \label{eq:v_global_lower}
\end{equation}
\rev{For the global certificate, $\Delta x_t^\ast$ is computed offline and $\mathcal{V}_{\text{global}}$ is evaluated online via a grid of 
resolution $\rho$.}


\subsection{Local Growth Iterative Refinement}
\label{subsec:local}

Since \(\mathcal{V}_{global}\) \rev{is an RA under-approximation, safe states may lie outside of it.}
\rev{We therefore introduce an online local refinement module, inspired by~\cite{lin2023generating}, that anchors a candidate ball $\mathcal{B}_r(\check{x}_t)$ to the nearest boundary point 
$\check{x}_t \in \partial\mathcal{V}_{global}$ to the current state and determines the 
maximum certifiable radius $r^\ast$ via a sample-and-shrink 
scenario optimization:}
\begin{equation}
\begin{aligned}
    r^\ast = & \max_{r \in \mathbb{R}} \quad  r \\
    & \text{s.t.} \quad  \bigwedge_{j=1}^N \sup_{t \in [0, T]} g_{\gamma}(\xi_{x_0^{(j)}}^{\pi_{lp}}, t) > 0, \quad \forall x_0^{(j)} \in \mathcal{B}_r(\check{x}_t).
\end{aligned} 
\label{eq:local_growth_ra}
\end{equation}

As shown in Algorithm \ref{alg:iterativegrowth}, the refinement 
centers an initial ball of radius $r$ on the nearest boundary point 
$\check{x}_t \in \partial\mathcal{V}_{global}$ (lines 1-2).
For \(M\) iterations, \(N\) initial conditions within the region are sampled within the ball and simulated under the \(\pi_{lp}\) (lines 3-4). If a violation of \eqref{eq:RAmeasure} is found, the radius 
shrinks to exclude it (lines 5-6), and
the process is repeated. If no violations are found, the ball 
$\mathcal{V}_{local} := \mathcal{B}_{r^\ast}(\check{x})$ is returned (lines 7-10).
By Theorem \ref{samplecomplexitythm}, the local certificate is:
\begin{equation}\label{eq:local growth refinement}
    \mathbb{P}_{S_2}(\mathbb{P}_{\mathcal{V}_{local}}( \sup_{t \in [0, T]} g_{\gamma}(\xi_{x_0^{(j)}}^{\pi_{lp}}, t) \leq 0 )\leq \epsilon_{local})\geq 1-\beta_{local}
\end{equation}
where \(\mathbb{P}_{S_2}\) denotes the probability measure associated with the samples, and 
\rev{the guarantee is with respect to the sampling distribution stating that}
the probability of failing to safely reach the target set is bounded by \(\epsilon_{local}\) with high confidence.

\noindent \rev{\textit{Remark (Construction vs. Verification):} If Algorithm~\ref{alg:iterativegrowth} breaks at iteration $\tau \leq M$, we regard iterations $0,\dots,\tau -1$ as a \textit{construction} phase that adaptively proposes a candidate radius, and iteration $\tau$ as a \textit{verification} phase: a single scenario optimization instance on a fresh i.i.d. batch drawn against the fixed $r^\ast$.
Certifying the full construction-then-verification pipeline is left to future work.}


\begin{algorithm}[t]
\caption{IterativeGrowth}
\label{alg:iterativegrowth}
\begin{algorithmic}[1]

\REQUIRE ${x}_t$, $\pi_{lp}$, $\mathcal{V}_{global}$, $N$, $r_{max}$, $M$, $T$
\ENSURE $r_{safe}$, $\mathcal{V}_{local}$,

\STATE \textbf{Initialize:} $r^\ast \gets r_{max}$
\STATE Find closest boundary point $\check{x} \in \partial \mathcal{V}_{global}$ to $x_t$

    
\FOR{iter=$0,1,\dots,M-1$}
    \STATE \(\mathcal{S}_i \gets\) Sample \(N\) states IID from \(\mathcal{B}_r(\check{x}_t)\), roll out trajectories \(\xi_{x}^{\pi_{lp}}\).
    \IF{\(\exists x \in \mathcal{S}_i\) s.t. \(\max_{t \in [0,T]} g_{\gamma}(\xi_{x}^{\pi_{lp}},t) \leq 0\)}
        \STATE \(r^\ast \gets \min_{x} \|x - \check{x}_t\|_2\) \\
        \quad \quad \textbf{s.t.} \(\max_{t \in [0,T]} g_{\gamma}(\xi_{x}^{\pi_{lp}},t) \leq 0\)
    \ELSE
        \STATE \textbf{return} $\mathcal{V}_{local} := \mathcal{B}_{r^*}(\check{x}_t)$
    \ENDIF
\ENDFOR


\end{algorithmic}
\end{algorithm}

\subsection{Hierarchical Verification and Control Framework} \label{subsec:hierarchical}
Our verification framework features a \rev{condition-based} hybrid switching logic that prioritizes the highest available tier of probabilistic certification.
\rev{When the state lies within a certified region, \(\pi_{lp}\) is executed; otherwise, control authority is delegated to a model-based fallback, Model Predictive Path Integral (MPPI) \cite{williams2016aggressive} warm-started with \(\pi_{lp}\).}
Since optimal RA policies may not guarantee target set invariance~\cite{chenevert2024solving}, MPPI is also used once \(x_t \in \mathcal{T}\).
The hierarchical architecture is illustrated in Fig.~\ref{fig:cover_fig} and Algorithm~\ref{alg:adaptive_hybrid}\rev{, detailing four modes}:
\begin{enumerate}
    \item \textbf{Tier 0: Target Maintenance Mode:} If \(x_t \in \mathcal{T}\), MPPI maintains the system within $\mathcal{T}$ \rev{(reach-avoid-and-stay)}.
    \item \textbf{Tier 1: Global Verification Mode:} If \(x_t \in \mathcal{V}_{\text{global}}\), \rev{the system is within the offline scenario-certified region, and} 
    \(\pi_{lp}(x_t)\) is executed.
    \item \textbf{Tier 2: Local Refinement Mode:} 
    If \(x_t \notin \mathcal{V}_{global}\) \rev{and} \(x_t \in \mathcal{V}_{local}\),
    the online local certificate confirms \(\pi_{lp}\) is reliable beyond the global boundary and is executed.
    \item \textbf{Tier 3: Safety Recovery Mode:} 
    If \(x_t \notin \mathcal{V}_{global}\) \rev{and} \(x_t \notin \mathcal{V}_{local}\), MPPI warm-started by $\pi_{lp}$ minimizes distance to 
    the nearest verified set boundary, steering the system into a 
    certified region.
\end{enumerate}

\rev{\subsection{Parameter Selection and Trade-offs}}
\rev{The framework is parameterized by $N,M,T$, and $\rho$. 
Sample size $N$ follows Theorem~\eqref{samplecomplexitythm}: smaller $\epsilon$, $\beta$ 
increase $N$, yielding a more conservative $\mathcal{V}_{\text{global}}$ 
and $\mathcal{V}_{\text{local}}$. 
The iteration cap $M$ bounds online sample-and-shrink, giving 
Tier~2 complexity $\mathcal{O}(M N  T)$. 
Horizon $T$ trades off 
error against rollout cost: larger $T$ reduces finite-time 
approximation error but increases computation. 
A smaller grid resolution \(\rho\) of the global certificate reduces conservatism but increases computation and should be chosen subject to real-time budget constraints.}
\begin{algorithm}[t]
\caption{Hierarchical Verification Framework}
\label{alg:adaptive_hybrid}
\begin{algorithmic}[1]

\REQUIRE $x_t$, $\mathcal{V}_{global}$, $\pi_{lp}$
\ENSURE $u_t$

\STATE $\mathcal{T} \gets \text{UpdateSet}(x_t)$

\IF{$x_t \in \mathcal{T}$}
    \STATE \textbf{return} $\text{MPPI\_Fast}(x_t)$  // Tier 0
\ELSIF{$x_t \in \mathcal{V}_{global}$}
    \STATE \textbf{return} $\pi_{lp}(x_t)$ // Tier 1
\ELSE
    \STATE $\mathcal{V}_{local} \gets \text{IterativeGrowth}(x_t,\pi_{lp},\mathcal{V}_{global})$

    \IF{$x_t \in \mathcal{V}_{local}$}
        \STATE \textbf{return} $\pi_{lp}(x_t)$ // Tier 2
    \ELSE
        \STATE \textbf{return} $\text{MPPI\_Rec}(\text{Warmstart}=\pi_{lp})$ // Tier 3
    \ENDIF
\ENDIF
\end{algorithmic}
\end{algorithm}

\section{Theoretical Reach-avoid Guarantees}


\subsection{Tier 1 \& Tier 2: Probabilistic RA Set Verification}
Given the hierarchical nature of the framework, we consider sequential probabilistic guarantees leveraged as needed.

\begin{proposition}[Sequential Probabilistic Guarantees]
\label{sequentialprop}
    Let $\mathcal{P}_{global}$ and $\mathcal{P}_{local}$ denote the global and local scenario sampling distributions, respectively.\\
    \textbf{Global certificate:} The global verified region $\mathcal{V}_{global}$ satisfies the following probabilistic certification:
\begin{equation*}
\begin{split}
    \mathbb{P}_{S_1}\Big(\mathbb{P}_{{\mathcal{P}_{global}}}\big(\{V_\gamma(x_0) < 0\} \, \cap & \\
    \{\check{V}_\gamma(\bar{x}_0, T) \geq 0 \}\big) \leq \epsilon_{global}\Big) & \geq 1-\beta_{global}.
\end{split}
\end{equation*}
\textbf{Local certificate:} For states sampled from the local refinement distribution,
\begin{equation*}
    \mathbb{P}_{S_2}(\mathbb{P}_{\mathcal{P}_{local}}( \sup_{t \in [0, T]} g_{\gamma}(\xi_{x_0^{(j)}}^{\pi_{lp}}, t) \leq 0 )\leq \epsilon_{local})\geq 1-\beta_{local}
\end{equation*}
\end{proposition}
\begin{proof}
This follows from Theorems~\ref{samplecomplexitythm} and~\ref{globalsafetythm}.
\end{proof}
\noindent\textit{Remark (Interpretation):} 
    $\epsilon_{global}$ and $\epsilon_{local}$ bound violation probabilities under $\mathcal{P}_{global}$ and $\mathcal{P}_{local}$, not the closed-loop runtime distribution. These certificates quantify the reliability of $\pi_{lp}$ under the verification sampling models and are used to probabilistically guide decision-making. We demonstrate the usefulness of certificate-guided decision-making during runtime in our experiments in Section~\ref{experiments}.

\begin{figure*}[t]
    \centering
    \begin{minipage}[c]{0.73\linewidth}
        \includegraphics[width=\linewidth]{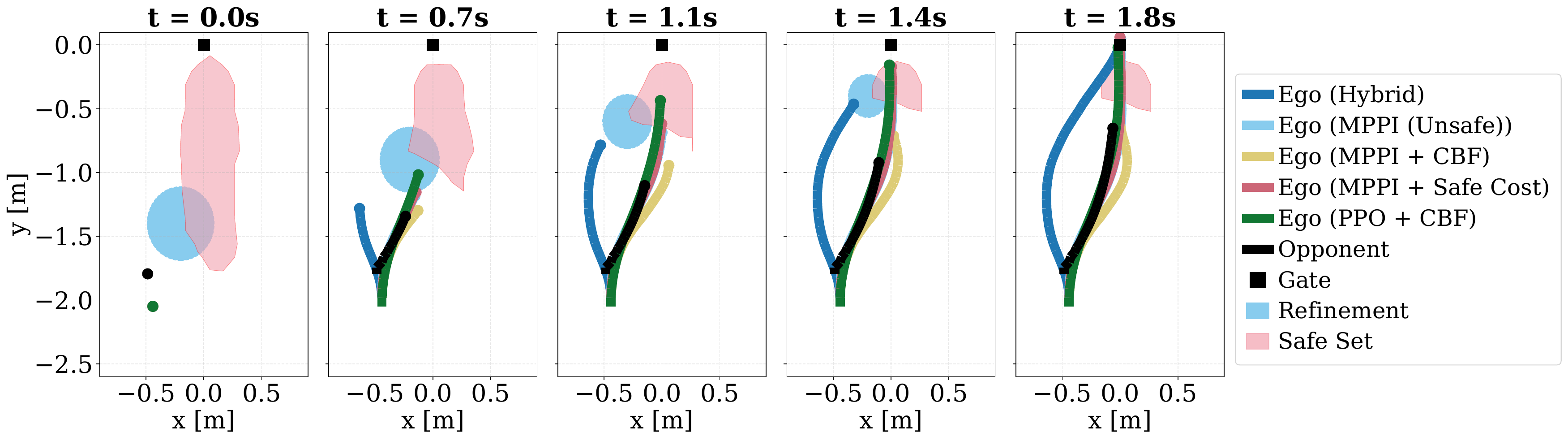}
    \end{minipage}
    \hfill
    \begin{minipage}[c]{0.26\linewidth}
        \caption{Evolution of a drone racing simulation against an opponent (black), comparing our method (Hybrid, dark blue) with various baselines. Baselines struggle to balance safety and goal-reaching, whereas our method achieves both, safely overtaking and reaching the gate. Moreover, we observe that the ego agent leverages safe recovery MPPI to try to enter the locally refined set.}
        \label{fig:trajectory}
    \end{minipage}
\end{figure*}

\begin{figure}[t]
    \centering
    \includegraphics[width=1.0\linewidth]{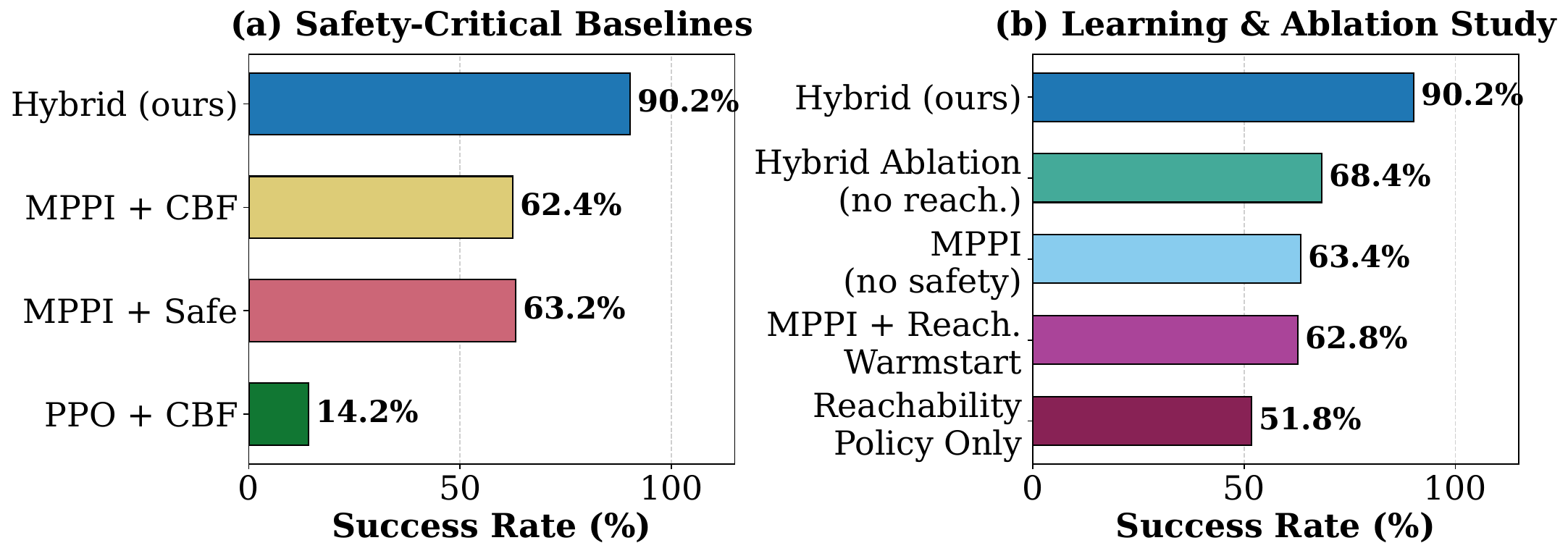}\vspace{-0.5em}
    \caption{(a) Performance of our framework compared to control barrier function (CBF) and soft-constraint baselines. (b) Ablation study comparing the hybrid method to its learning-only and model-based variants. In both (a) and (b), the results for our method are shown for reference. \emph{Our framework achieves higher success probability over 500 sampled initial conditions, effectively balancing collision avoidance and goal-reaching.}
    }
    \label{fig:bar_graphs}
\end{figure}

\subsection{Tier 3: Finite-time Recovery}
If \(x_t \notin (\mathcal{V}_{global} \cup \mathcal{V}_{local}) \), no safety guarantees can be made and recovery via reachability-warm-started MPPI is invoked. Let \(\mathcal{V}_{recover}\) denote the forward reachable set of the ego agent at \(x_t\). Then, recovery is feasible when 
$\mathcal{V}_{recover}$ satisfies $\mathcal{V}_{recover} \cap 
(\mathcal{V}_{global} \cup \mathcal{V}_{local}) \neq \varnothing$.


\section{Experiments and Results}
\label{experiments}
\rev{We evaluate our framework on a high-dimensional drone-racing task. 
Given the high dimensionality of the state, fine-grained, grid-based verification baselines are
intractable. Therefore, comparison against model-based and 
learning-based controllers is the standard benchmarking 
approach.}

\subsection{Drone Racing Case Study}
We consider a drone racing example as defined in \cite{li2025certifiable}, a challenging high-dimensional testbed with dynamic interactions and safety-critical constraints, where the ego agent must safely overtake an opponent to pass through the gate first. For computational efficiency, we set \(\beta_{global}=\beta_{local} = 0.001\) and \(\epsilon_{global} = \epsilon_{local}= 0.1\) 
as the smallest values keeping $N$ within our online sample budget, and so here $N = 158$.
\subsubsection{Joint System Dynamics}
\label{jointdynamics}
We consider a 12D non-cooperative racing scenario involving an ego drone (\(e\)) and an opponent drone (\(o\)). The state of each drone \(i \in \{e, o\}\) is defined by its 3D position and velocity: \({x}^i_t=[p_{x,t}^i, v_{x,t}^i, p_{y,t}^i, v_{y,t}^i, p_{z,t}^i, v_{z,t}^i] \in \mathbb{R}^6\). The joint system state is defined as the concatenation \({x}_t=[{x}^e_t, {x}^o_t] \in \mathcal{X} \subseteq \mathbb{R}^{12}\). 
The \(i\)-th drone is modeled by double-integrator dynamics, with control input \({u}_t^i = [a^i_{x,t}, a^i_{y,t}, a^i_{z,t}]\), satisfying \(\norm{{u}^i_t}_\infty \leq \epsilon_u := 1\,\text{m/s}^2\). We assume the opponent has bounded control authority and follows an LQR controller that drives it toward the gate aggressively, with high speed. 

\subsubsection{Target Set} The mission objective is to pass through a narrow gate centered at \([0,0,0]\) with width \(w\) and height \(h\). However, for the purpose of reachability, the target set for the ego drone is defined as
\begin{equation}
\label{eq:targetset}
    \mathcal{T}=\left\{ 
    \begin{array}{lcr}
       x:p_y^e-p_y^o >0,  & v_y^e-v_y^o>0,  \\
       |p_x^e|<0.3,  & |p_z^e|<0.3 
    \end{array}
     \right\}.
\end{equation}
This target set requires the ego drone to be ahead of and faster than the opponent while staying within a corridor for gate passage (radius 0.3 m). 

\subsubsection{Safety Constraints}
To ensure safe flight, the ego drone should avoid the area affected by the downward airflow from the opponent drone, leading to the following nonconvex constraint, which is challenging to satisfy:
\begin{equation}
\label{eq:constraints}
    \norm{
    \left[
    \begin{array}{cc}
         p^e_{x,t}-p^o_{x,t}  \\
         p^e_{y,t}-p^o_{y,t} 
    \end{array}
    \right]
    }^2_2> \Big(1+\max(p^o_{z,t}-p^e_{z,t},0)\Big) \times 0.2,
\end{equation}
where the collision avoidance region between the ego and opponent drone grows as their height difference increases. Moreover, to ensure the ego drone passes through the gate without colliding with its boundaries, we impose the following constraints as in \cite{li2025certifiable}:
\begin{equation}
    \pm p_{x,t}^e - p_{y,t}^e  > -0.05,\hspace{0.5cm}  \pm p_{z,t}^e - p_{y,t}^e  > -0.05.
\end{equation}
\subsection{Hypothesis 1: The reachability-based guidance of our method leads to a higher success rate than soft- and hard-constrained model-based controllers or RL policies.}
\label{hypothesis1}

We empirically compare our framework to various safe baseline methods: MPPI with a control barrier function (CBF) safety filter, MPPI with safe soft constraints, and proximal policy optimization (PPO) with a CBF safety filter. We sample 500 initial conditions in a bounded region of the state space and roll out the trajectories. Success is quantified by the ego safely reaching the gate before the opponent, and our framework achieves a higher rate of success than its safe counterparts as shown in Figure \ref{fig:bar_graphs}(a). Qualitatively, baseline methods tend to prioritize goal-reaching at the expense of safety, whereas our method demonstrates more robust collision avoidance during overtaking (Fig.~\ref{fig:trajectory}).

\subsection{Hypothesis 2: The hierarchical control framework features a better combination of MPPI and learning.}
Similar to Section \ref{hypothesis1}, we conduct an empirical study comparing our framework to baseline methods: Hybrid Ablation which conducts switching but replaces the reachability policy with MPPI, pure MPPI with no safety consideration, MPPI warm-started with the reachability policy, and the reachability policy only. We select these baselines to observe the interaction between MPPI, the reachability policy, and the switching mechanism. 
We sample 500 initial conditions in a bounded region of the state space and roll out the trajectories. 
The results in Figure \ref{fig:bar_graphs}(b) highlight the necessity of both the switching mechanism and the reachability policy as none of the other baselines achieved the same rate of success as our hierarchical framework. Notably, the reachability policy prioritizes reaching the target set but does not guarantee maintaining the system within it, a limitation inherent to reachability theory, as also observed in \cite{chenevert2024solving}. This explains its lower success rate when used alone.

\section{CONCLUSIONS}
In this work, we propose an efficient hierarchical probabilistic verification framework for reachability learning that unifies global, local, offline, and online scenario-based certification. By combining a coarse global certificate via scenario optimization with an online local refinement, the approach reduces conservatism while focusing computation on critical regions of the state space. A condition-based switching mechanism \rev{integrates the learned reachability policy with a model-based MPPI controller, enabling adaptive decision-making when certified regions are unavailable. The resulting framework provides scenario-based probabilistic certificates that quantify the reliability of the learned policy under the sampling distributions used in verification. These certificates are used to guide hierarchical control switching.}
Future work includes hardware validation and extension to more complex multi-agent settings and uncertain real-world environments.
\section*{APPENDIX}
\subsection{Runtime Statistics}
\label{app:runtime}

All experiments were conducted on a workstation running Ubuntu~24.04 
with an AMD Ryzen~9 9950X CPU (128~GB RAM) and an NVIDIA RTX~4080 GPU (16~GB VRAM, CUDA~12.9). Computational overhead is decoupled into offline initialization (one-time) and online execution (per control step, $\Delta t = 0.1$s). The global certificate $\mathcal{V}_{global}$ uses $N=158$ scenario samples (Theorem~\ref{samplecomplexitythm} with $\epsilon=0.1$, $\delta=10^{-3}$). Tier~1 evaluates the learned value function over a 468-point grid ($18{\times}26$, resolution 0.1~m) in the ego $(x,y)$ subspace conditioned on the current opponent state via a single batched GPU rollout of $468{\times}30$ steps. The high variance in Tier~1 cost reflects online re-certification at each step as the opponent state evolves; if sufficient memory permits offline 
precomputation of the full 12D certificate, Tier~1 reduces to a single interpolation query.

\begin{table}[h]
\centering
\caption{Computational cost and execution statistics (12D quadrotor, 
         50 trajectories, 1{,}423 steps).}
\label{tab:computational_overhead}
\begin{tabular}{lccc}
\toprule
\textbf{Phase / Tier} & \textbf{Mean} & \textbf{Std} & 
\textbf{Active} \\ 
\midrule
\multicolumn{4}{l}{\textit{Offline (one-time)}} \\
$\quad$ Policy training ($\pi_{lp}$) & 30 min & — & — \\
$\quad$ Global cert.\ ($\mathcal{V}_{global}$, $N{=}158$) & $<$1 s & — & — \\ 
\midrule
\multicolumn{4}{l}{\textit{Online (per step, $\Delta t=0.1$s)}} \\
$\quad$ Tier 1: 2D grid re-cert.\ & 77.1 ms & 90.9 ms & 46\% \\
$\quad$ Tier 2: local refinement & 47.3 ms & 38.0 ms & 47\% \\
$\quad$ Tier 3: MPPI fallback & 158.4 ms & 2.0 ms & 38\% \\
\midrule
$\quad$ \textbf{Full solve()} & \textbf{125.6 ms} & \textbf{143.3 ms} & 100\% \\ 
\bottomrule
\end{tabular}
\end{table}






\bibliographystyle{ieeetr}
\bibliography{references}

\end{document}